\documentclass[10pt]{article}

\pdfoutput=1

\usepackage{upgreek}
\usepackage[greek,english]{babel}
\usepackage{amsmath}
\usepackage{amssymb}
\usepackage{amsfonts}
\usepackage{amsthm}
\usepackage{graphicx}
\usepackage{mathtools}
\usepackage{txfonts}
\usepackage{mathrsfs}
\usepackage{tikz}
\usetikzlibrary{decorations.markings,fadings}
\usepackage[all]{xy}
\usepackage{color}
\usepackage{accents}

\newcommand{\intl}[2]{\!\underset{#1}{\overset{#2}{\rotatebox[origin=rc]{15}{\large\ensuremath{\int}}}}}

\newcommand{\im}{\operatorname{Im}}

\newcommand{\codim}{\operatorname{codim}}
\newcommand{\id}{\operatorname{Id}}
\newcommand{\rank}{\operatorname{rank}}

\renewcommand{\alpha}{\alphaup}

\renewcommand{\beta}{\betaup}
\renewcommand{\gamma}{\gammaup}
\renewcommand{\delta}{\deltaup}

\renewcommand{\epsilon}{\varepsilonup}
\renewcommand{\varepsilon}{\epsilonup}

\renewcommand{\zeta}{\zetaup}

\renewcommand{\eta}{\etaup}
\renewcommand{\theta}{\thetaup}
\renewcommand{\vartheta}{\varthetaup}

\renewcommand{\iota}{\iotaup}
\newcommand{\Kappa}{\mathrm{K}}
\renewcommand{\kappa}{\varkappa}
\renewcommand{\lambda}{\lambdaup}

\renewcommand{\mu}{\muup}

\renewcommand{\nu}{\nuup}
\renewcommand{\xi}{\xiup}

\renewcommand{\pi}{\piup}
\newcommand{\Rho}{\mathrm{P}}
\renewcommand{\rho}{\rhoup}
\renewcommand{\varrho}{\varrhoup}
\renewcommand{\sigma}{\sigmaup}
\renewcommand{\varsigma}{\varsigmaup}

\renewcommand{\tau}{\tauup}
\renewcommand{\Upsilon}{\textrm{\greektext U}}
\renewcommand{\upsilon}{\upsilonup}
\renewcommand{\phi}{\upvarphi}
\renewcommand{\varphi}{\phiup}

\renewcommand{\chi}{\chiup}
\renewcommand{\psi}{\textrm{\greektext y}}
\renewcommand{\omega}{\omegaup}

\renewcommand{\mathbb}{\varmathbb}

\tikzset{->-/.style={decoration={
  markings,
  mark=at position #1 with {\arrow{>}}},postaction={decorate}}}

\title{
Allostery and conformational changes upon binding as generic features of proteins: a high-dimension geometrical approach.
}

\author{A.S. Zadorin}
\date{\small\itshape
	Chimie Biologie Innovation, ESPCI Paris, CNRS, PSL University, 75005 Paris, France.\\
	Center for Interdisciplinary Research in Biology (CIRB), Coll\`ege de France, CNRS, INSERM, PSL
	Research University, Paris, France.}

\newtheorem{theorem}{Theorem}
\newtheorem{lemma}{Lemma}

\theoremstyle{definition}
\newtheorem*{definition}{Definition}

\begin{document}

\maketitle

\begin{abstract}

A growing number of experimental evidence shows that it is general for a ligand binding protein to have a potential for allosteric regulation and for 
further evolution. In addition, such proteins generically change their conformation upon binding. O. Rivoire has recently proposed an evolutionary 
scenario that explains these properties as a generic byproduct of selection for exquisite discrimination between very similar ligands. The initial 
claim was supported by two classes of basic examples: continuous protein models with small numbers of degrees of freedom, on which the development of a 
conformational switch was established, and a 2-dimensional spin glass model supporting the rest of the statement. This work aimed to clarify the 
implication of the exquisite discrimination for smooth models with large number of degrees of freedom, the situation closer to real biological systems. 
With the help of differential geometry, jet-space analysis, and transversality theorems, it is shown that the claim holds true for any generic flexible 
system that can be described in terms of smooth manifolds. The result suggests that, indeed, evolutionary solutions to the exquisite discrimination 
problem, if exist, are located near a codimension-1 subspace of the appropriate genotypical space. This constraint, in turn, gives rise to a potential 
for the allosteric regulation of the discrimination \emph{via} generic conformational changes upon binding.
 
\end{abstract}

\section{\large Introduction}

Significant conformational changes in proteins upon their binding to specific ligands are ubiquitous in nature. Their occurrence spans from signaling 
proteins and transcription factors to enzymes. They are observed even outside the realm of proteins: in aptamers and ribozymes. Allosteric regulation, 
the modulation of the primary function by binding of another molecule at a distant site, is often associated with such biopolymers either in an actual 
or in a potential form. There are two common features in all these systems. First, they are flexible. Second, the primary task that all of them solve 
includes a fine discrimination between different but close ligands vs. solvent. Namely, the desirable ligand must be more preferably than the solvent 
with the contrary for undesirable ligands. For signaling proteins it is the ability to distinguish between the specific signal and similar molecules. 
For transcription factors it is the recognition site among similar DNA motifs. For enzymes and ribozymes it is the ability to sufficiently strongly 
bind the substrate and/or the transition state and to release the (often similar) product \cite{Hammes2009}.

The main focus of researchers since the discovery of conformational changes has been on the nature and mechanisms of these changes. Their evolutionary 
origin is usually seen either as a requirement for the function of the protein in question or as a subsequent development of regulation of this 
function. For example, for a signal transduction receptor, a conformational change upon binding is necessary to initiate the signal response pathway 
(be it binding to a DNA site or a transmembrane activation of a signaling cascade). For enzymes, such explanations include the correct positioning of 
aminoacids in the reaction center and creation of the correct environment around the substrate, as well as kinetic control of the reaction rate. For 
the enzymes that are molecular motors, conformational changes are the essence of their function. Finally, the conformational change is seen as an 
adaptation to allosteric regulation of protein function \cite{Hammes2002}. It is important to emphasize that in this view allostery is actively 
selected for and a conformational change serves as a means to fulfill this demand. Such explanations assume some adaptive value of the conformational 
change in light of selection for a complex property of the protein. The conformational changes themselves are understood as highly orchestrated events.

The main two hypothetical mechanisms of the conformational change itself, however, assume it to be generic. These two mechanisms are the hypothesis of 
\emph{induced fit} and the hypothesis of \emph{conformational selection}. In the induced fit scenario, the ligand provokes a conformational change in 
the sufficiently flexible binding protein after an initial weak binding. This results in a strongly bound complex \cite{Koshland1958,Koshland1995}. In 
the conformational selection paradigm, the native state and the conformation of the strongly bound complex exist as possible conformations in the 
population of the free protein under normal conditions. The native state is assumed to reflect the global minimum of the free energy while the other 
state is assumed to be metastable with lower probability in the population. The binding of the ligand change their roles and the other state becomes 
predominant. The conformations themselves, in this scenario, do not significantly change---only their free energy levels change 
\cite{Monod1965,Changeux2011}.

Recently, a different view on the problem of allostery and conformational changes in proteins was introduced by O.~Rivoire in \cite{Rivoire2018}. 
Rivoire noticed that allostery can be a consequence of an existing conformational change in a discriminating protein. Indeed, if a conformational 
change is invelved in an exquisite discrimination, the ability to discriminate can be turned off by blocking the movement itself and thus changing the 
energies of bound states. If the conformation changes far from the initial binding site, as it takes place in sufficiently large conformational change, 
the regulating binding can happen far from the initial binding pocket. This situation is interpreted as allosteric regulation. The conformational 
change in the first place, in this scenario, comes as a byproduct of the selection towards the exquisite discrimination. Thus, a potential for 
allostery emerges from a selection for a much simpler property.

The validity of this scenario was demonstrated in \cite{Rivoire2018} on two types of models: 1) extremely simple elastic network models and 2) a spin 
glass model of proteins. The elastic models treated the protein as a single mass on one or two springs with only one degree of freedom. Possible 
ligands were treated as numerical values on an axis of \emph{environmental variables} that served as an additional force constantly acting on the 
system. In addition, the evolutionary degree of freedom also was considered to be a single continuous variable that imposes another force of the same 
sort. The spin glass model, from the other hand, had multiple configurational and evolutionary degrees of freedom, but a configuration of each 
``aminoacid'' was described only by either ``up'' or ``down'' state. Both these models are strongly simplified descriptions of proteins.

Both models gave the same qualitative result, formulated in \cite{Rivoire2018} as follows. Under the assumption of a system's flexibility, the 
discrimination of particular similar ligands requires the system to be evolutionary finely tuned to respect requirements on free energies of the 
complexes. This constraint causes a generic conformational change upon binding. A hypothesis, tested only on the spin glass model of proteins, was 
formulated that connects a large enough conformational change with the potential for allosteric regulation by involvement of distant parts of the 
protein in the movement. Being involved in keeping a delicate free energy balance, such sites become a potential target for further regulation by 
another ligand, for example. Furthermore, it was shown (on the continuous elastic model) that the conformational change may have a form of continuous 
deformation of the initial state to the final one or these states may be different ones and may even coexist as a global and a local minimum. Thus the 
distinction between the induced fit and the conformational selection becomes moot from this point of view.

However, the simplicity of the models used for illustration of this powerful principle comes with strong limitations that may prevent a direct 
generalization. The drawback of the spin glass model is in its intrinsic discontinuity, and it is difficult to say if the observed effects are related 
to general properties of protein-like systems or to this particularity of the model. This is especially true for the claim about a conformational 
switch, since the behaviour of the system is switch-like at the level of each element from the beginning.

The continuous elastic model suffers from its low dimensionality. It is not clear how the conclusion can be drawn from an example with one physical 
degree of freedom, one scalar phenotypic trait, and a ligand space described by a single number. Furthermore, the particularly simple relation between 
the ligand and the system's potential may turn out to be a very particular case.

In the current work, it is proven that the conclusions of \cite{Rivoire2018} are, indeed, valid, under a certain interpretation, for a much wider class 
of models: continuous systems with any number of degrees of physical freedom, any dimensionality of the phenotypical trait space, and any number of 
parameters describing ligands. In addition, an estimate on the abundance of evolutionary solutions to the exquisite discrimination of particular 
ligands (equivalently, on the required fine tuning of the protein sequence) is derived in terms of the dimensionality of the set in the trait space 
around which the solutions are concentrated. It is also shown that the proposed scenario for the origin of allosteric regulation is plausible in these 
settings, too.

The work is organized in the following way. In Section~\ref{section-physical}, the problem is formulated in terms of physical chemistry. In 
Section~\ref{section-mathematical}, it is translated to a mathematical model. In Section~\ref{section-theorems}, the problem is rigorously formalized 
in the language of differential geometry and three main theorems are stated, constituting the main result about the exquisite discrimination problem, 
the conformational changes, and allostery: Theorem~\ref{main}, \ref{main2}, and \ref{main3}. A biological interpretation and some implications of 
theses results are outlined in Section~\ref{section-discussion}. The final Section~\ref{section-proof} is completely devoted to formal mathematical 
proofs of the main theorems.

\section{\large Physical formulation of the problem}
\label{section-physical}

Following \cite{Rivoire2018}, we assume that evolution of a protein involves three types of variables: 1) physical (conformational) degrees of freedom 
$x$, 2) environmental degrees of freedom $\ell$ that define the surrounding medium, and 3) evolutionary degrees of freedom $a$ associated with the 
protein sequence.

Typically, the variables in $x$ include positions of single atoms or distances between pairs of them and completely describes the shape (conformation) 
of the molecule. Depending on the coarse graining of the model, it may describe mutual orientation of larger portions of the molecule, like individual 
aminoacids. In the latter case, $x$ may also involve angles on top of distances. The variables in $\ell$ contain information about the environment 
around the molecule. In this particular case they are restricted to the identification of a ligand bound to a particular site of the protein or to the 
absence of any such ligand (a molecule of the solvent can be taken as the ligand for this case). Finally, $a$ describes the genetic information 
involved in building the molecule. In the most direct case it is its aminoacid sequence. Alternatively, it can reflect some higher level aggregated 
phenotypical properties of the molecule or its parts that define its behaviour in the selected level of abstraction.

These three types of variables are linked \emph{via} the parametrized potential energy $U(x,\ell,a)$ of the protein, where $\ell$ and $a$ are 
parameters. Thus, a function $U(x,\ell,a)$ describes a family of potential energies $U_{\ell,a}(x)$ with a constant parameter $a$ (it defines the 
protein) and an environment-dependent parameter $\ell$ (it describes how the energy changes with the binding of the ligand $\ell$). At given $\ell$ and 
$a$, the distribution of conformations of the protein is given by the Boltzmann distribution with that potential energy such that the probability of 
conformation $x$ in an ensemble of molecules is $\mathbb P(x \mid \ell,a) = \exp\Big(-\beta\big(U(x,\ell,a) - F(\ell,a)\big)\Big)$, where $F(\ell,a)$ 
is the free energy of the system and $\beta$ is the inverse temperature measured in energy units. Following the treatment of the continuous case in 
\cite{Rivoire2018}, we will only consider the zero temperature limit $\beta \to \infty$. In this case $\mathbb P(x \mid \ell,a)$ degenerates to a 
$\delta$-function at the point (or points) of the global minimum of $U_{\ell,a}$ and $F(a,\ell)$ is equal to this minimum.

For example, in \cite{Rivoire2018}, variables $x$, $\ell$, and $a$ were natural numbers and the considered potentials had the following forms
	\begin{align}
		&U(x,\ell,a) = \frac{1}{2}k(|x| - r)^2 - (\ell - a)x,\notag\\
		&U(x,\ell,a) = \frac{1}{2}k(x - r)^2 - (\ell - a)x,\\
		&U(x,\ell,a) = k\left(\sqrt{x^2 + d^2} - r\right)^2 - (\ell - a)x,\notag
	\end{align}
where $k$, $d$, and $r$ are positive constants.

The \emph{exquisite discrimination problem} is formulated in the following way. Given a desirable ligand $\ell_r$ and an undesirable ligand 
$\ell_\varw$ such that $\ell_r \approx \ell_\varw$, and assuming that the environment defined by the solvent alone is represented by 
$\ell_\varnothing$, find $a$ such that
	\begin{equation}
		F(\ell_r,a) < F(\ell_\varnothing,a) < F(\ell_\varw,a),
	\end{equation}

or, in the zero temperature limit,
	\begin{equation}
		\min U_{\ell_r,\,a} < \min U_{\ell_\varnothing,\,a} < \min U_{\ell_\varw,\,a}.
	\label{U-condition}
	\end{equation}

\noindent This condition is schematically shown on the left part of Figure~\ref{scheme}.

\section{\large Mathematical formulation of the problem}
\label{section-mathematical}

For the sake of brevity, we will use the term ``protein'' for ``a system that needs to do an exquisite discrimination'', although, of course, the 
general argument is not restricted only to proteins. We will assume that a protein is characterized by three types of variables: conformational 
variables $x$ that take value in the configuration space $X$, environmental variables $\ell$ taking value in $L$ (the space of possible ligands), and 
evolutionary variables $a$ taking value in $A$ (the space of protein sequences, the phenotypical trait space, etc.).

The main claim of the current work can informally be expressed in the following statement. \emph{Under general assumptions, a sufficiently flexible 
system that solves the exquisite discrimination problem generically experiences a large conformational change upon binding to its substrate. The 
ability to discriminate requires evolutionary fine tuning and possible solutions are concentrated near a codimension-1 hypersurface of the phenotypical 
trait space $A$. The combination of the fine tuning and the conformational change makes the discrimination ability sensitive to binding of other 
ligands to distant sites.}

To make these statement precise we will fix the following assumptions. Spaces $X$, $L$, and $A$ are assumed to be smooth ($C^\infty$) compact 
manifolds. The physical behaviour of a protein is defined by an energy function $U\colon M \to \mathbb R$, where
	\begin{equation}
		M = X\times L\times A.
	\end{equation}

\noindent The product is understood in the category of smooth manifolds so $M$ is assumed to be endowed with the structure of a $C^\infty$ manifold. 
$U$ is assumed to be smooth, as well. $U$ is understood as a family of potentials on $X$ with parameters from $L$ and $A$. In the zero temperature 
approximation, the configuration of a protein with sequence $a$ that corresponds to a ligand $\ell$ is considered to be $x$ that minimizes 
$U(x,\ell,a)$ with constant $\ell$ and $a$ globally in $X$. We also assume that $X$ represents only the \emph{shape} of the protein and the degrees of 
freedom of the whole molecule (translational and rotational) are already excluded as well as that the dimensionality of $X$ corresponds to the number 
of the leftover \emph{independent} degrees of freedom.

Let us discuss these assumptions. The representation of the configuration space of a physical system by a manifold is very natural and does not require 
any special explanation. The compactness of $X$ is a technical requirement, which is not very restrictive. Indeed, if the configurations are given by 
the collection of pairwise distances between elements (aminoacids, nucleotides) with hard links between neighbours in the primary sequence, as it is 
commonly assumed in physical models of macromolecules, the configuration space naturally has a form of closed (without border) compact multiply 
connected submanifold of some Euclidean space. If, instead, no restrictions are applied, but the interatomic interactions are described by some 
pairwise potentials, the actual configuration space is some Euclidean space, which is not compact. However, as the interaction potential either 
increases with acceleration (as in elastic network models) or monotonously increase to some finite limit (as in real molecules) as some atom approaches 
an infinite distance from the rest of the structure, we are not interested at the behaviour of $U$ in the neighbourhood of infinity. In this case, it 
is enough to consider some smaller, compact, subspace of the initial configuration space. Finally, when there is a finite potential energy in the bound 
state of the protein(as in real molecules), the initial space $\mathbb R^n$ can be compactified to the projective space $P\mathbb R^n$. The compactness 
of $L$, as well as its smooth nature, is just a convenient hypothesis as there are no good models for this space. The sequence space $A$ is not a 
smooth manifold in nature. It is rather a nondirected graph with high symmetry. However, we assume that it can be well approximated by some smooth 
manifold with a smooth function $U(x,\ell,a)$. For example, for binary sequences of length $n$ the sequence space is an $n$-dimensional hypercube. It 
can be approximated by an $(n-1)$-dimensional hypersphere.

We assume that for the protein to perform its function, it must bind the correct ligand, described by the environment $\ell_r$, stronger than the 
solvent, $\ell_\varnothing$, and the incorrect ligand, $\ell_\varw$, must be bound weaker than the solvent. $\ell_r$ and $\ell_\varw$ are assumed to be 
close in $L$ ($L$, being described by some physical parameters, is usually metrizable, so one may assume that the distance between the states given by 
some metrics on this space is much smaller than between either of them and $\ell_\varnothing$).

As the first step, we will solve a simpler problem. Given two ligands $\ell_0,\ell_1\in L$, we want to find such phenotypes $a$ and such configurations 
$x$ that the protein bound to either of the ligands has the same minimal energy $U$. This point of view ignores the inevitably small differences 
between the minimal energy levels of the complexes with $\ell_r$ and $\ell_\varw$. Moreover, the values of $\ell_\varw$ and $\ell_r$ are considered to 
be indistinguishable, too. Therefore, we assume $\ell_\varw = \ell_r = \ell_1$. As a consequence, the minimal energy that corresponds to the binding to 
$\ell_\varnothing = \ell_0$ is equal to that of $\ell_1$, as we assume it to be between $\ell_\varw$ and $\ell_r$ (see Figure~\ref{scheme}).

This simplification can be regarded as a coarse grained view on the problem, where any slightly different points in any of the spaces are seen as 
equal. In this picture, a significant difference between real configurations $x_1$ and $x_2$ means simply $x_1 \neq x_2$. Such abstraction allows a 
rigorous mathematical treatment as it can be recast to questions about intersections of submanifolds of appropriate manifolds. Such questions, taking 
into account the notion of general position (generic case), result in exact answers in qualitative terms. The backinterpretation is, however, much less 
rigorous. We will address the issues related to it in the end of the article.

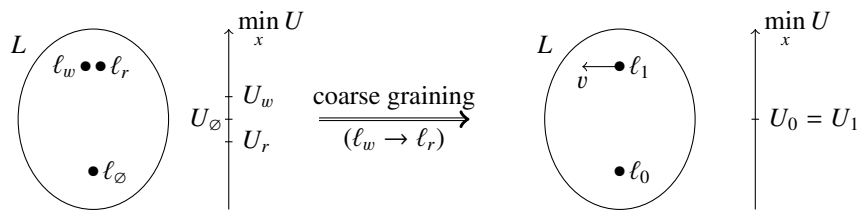
\begin{figure}
\begin{center}
		\begin{tikzpicture}
			\begin{scope}[xshift=0,yshift=0]
				\draw (0,0) ellipse (1cm and 1.2cm);
				\node at (-1,1) {$L$};
				\coordinate (o) at (0,-0.7);
				\coordinate (w) at (-0.1,0.7);
				\coordinate (r) at (0.1,0.7);
				\node at (o) {$\bullet$};
				\node[right] at (o) {$\ell_\varnothing$};
				\node at (w) {$\bullet$};
				\node[left] at (w) {$\ell_\varw$};
				\node at (r) {$\bullet$};
				\node[right] at (r) {$\ell_r$};
				\def\x{1.8};\def\dx{0.05};
				\draw[->] (\x,-1.2) -- (\x,1.2) node[right]{$\min\limits_x U$};
				\draw (\x-\dx,0) -- (\x+\dx,0) node[left]{$U_\varnothing$}; 
				\draw (\x-\dx,-0.3) -- (\x+\dx,-0.3) node[right]{$U_r$}; 
				\draw (\x-\dx,0.3) -- (\x+\dx,0.3) node[right]{$U_\varw$};
				\draw[double,->] (3,0) --node[above]{coarse graining}node[below]{($\ell_\varw \to \ell_r$)} (5,0);
			\end{scope}
			\begin{scope}[xshift=7cm,yshift=0]
				\draw (0,0) ellipse (1cm and 1.2cm);
				\node at (-1,1) {$L$};
				\coordinate (o) at (0,-0.7);
				\coordinate (r) at (0,0.7);
				\node at (o) {$\bullet$};
				\node[right] at (o) {$\ell_0$};
				\node at (r) {$\bullet$};
				\node[right] at (r) {$\ell_1$};
				\def\x{1.8};\def\dx{0.05};
				\draw[->] (\x,-1.2) -- (\x,1.2) node[right]{$\min\limits_x U$};
				\draw (\x-\dx,0) -- (\x+\dx,0) node[right]{$U_0 = U_1$};
				\draw[->] (r) -- (-0.5,0.7)node[below]{$\varv$};
			\end{scope}
		\end{tikzpicture}
\end{center}

\caption{
		A schematic representation of the coarse graining approach to build the mathematical formalization of the problem. Here $L$ is an abstract 
		ligand space, $\ell_r$ is the correct ligand, $\ell_\varw$ is the incorrect lignad, $\ell_\varnothing$ symbolises the solvent or the empty 
		binding site, $U$ is the potential energy of the system.
	}
\label{scheme}
\end{figure}

The initial problem sought a phenotype $a \in A$ that brings minimal energies for the ligand $\ell_\varnothing$ to that of the ligands $\ell_\varw$ and 
$\ell_r$ (with the correct ordering, but we will consider this issue separately). The corresponding protein would be considered to take a large 
conformational change upon binding if the corresponding configurations $x_\varnothing$ and $x_\varw \approx x_r$ are very different. In the simplified 
problem, the initial discrimination problem reduces to finding $a$ such that the global minima of $U$ in $X$ corresponding to $\ell_0$ and $\ell_1$ 
have the same energy level. We will call this a \emph{reduced discrimination problem}. Then the initial question is whether or not the corresponding 
minimum points $x_0$ and $x_1$ coincide in $X$.

We will also consider an \emph{infinitesimal discrimination problem}, where $\ell_r \approx \ell_\varw$ and we replace the difference between 
$\ell_\varw$ and $\ell_r$ by a vector $\varv$ in $L$ at $\ell_1$ that shows the direction from $\ell_r$ to $\ell_\varw$. We will assume that a given 
$a$, which solves the reduced discrimination problem, also solves the infinitesimal discrimination problem if the displacement along $\varv$ in $L$ at 
$\ell_1$ corresponds to a positive change in the energy value at the global minimum (see Figure~\ref{scheme}).

\section{\large Main results}
\label{section-theorems}

\subsection{\normalsize Exquisite discrimination, conformational changes, and fine tuning}

Let us recall the following notion. 

\begin{definition} A subset of a topological space is called a \emph{residual set} if it can be represented by a countable intersection of open dense 
subsets. A \emph{typical} element of the topological space is an element that belongs to some residual set. A situation is \emph{generic} if it can be 
represented as a typical element of some space relevant to the problem. A complement to a residual set is called \emph{meager set}. \end{definition}

We must consider that the naturally defined system, given by $U$, is typical. Indeed, the meaning of residual sets is that their complements, meager 
sets, can be considered as negligible and points that belong to them as special. An assumption that naturally occurring systems do not belong to some 
negligible sets is a kind of an extension of the Copernican principle. It is in this sense $U$ is typical.

The first main result now can be formulated in the following theorem.

\begin{theorem}
For a typical family of potentials $U \in C^\infty(M)$, solutions to the reduced discrimination problem for ligands $\ell_0$ and $\ell_1$ either do not 
exist, or a typical solution is located on a $(\dim A - 1)$-dimensional submanifold $\hat {\normalfont \Upsilon}$ of $A$ and, if $\dim X > 0$, its 
minimum points $x_0$ (for $\ell_0$) and $x_1$ (for $\ell_1$) are different.
\label{main}
\end{theorem}

It should be noted that this mathematical result is intuitively expected from the beginning. Indeed, one formally has to find points of minimum for 
$U_{\ell_0,a}$ and $U_{\ell_1,a}$. Let them be $\hat x_0(a)$ and $\hat x_1(a)$, respectively, where the dependence on $a$ is explicitly indicated. The 
solution is given by traits $a$ such that $U(\hat x_0(a), \ell_0, a) = U(\hat x_1(a), \ell_1, a)$. This constitutes one condition on 
$a$, which is intuitively expected to be satisfied on a codimension-1 hypersurface of $A$. In the same way, additional constraint of no conformational 
change is written in the form $\hat x_0(a) = \hat x_1(a)$ and is equivalent to $\dim X$ additional conditions. One would intuitively expect that the 
set of solution to discrimination without conformational changes occupies a submanifold of codimension $\dim X + 1$. However, the intuition alone is 
not suitable to treat multidimensional problems. In particular, the condition on $a$ is not a simple equation but depends on solutions $\hat x_0$ and 
$\hat x_1$ of the energy minimization problem. These solutions themselves depend on $a$ in a complex manner, which may involve discontinuities of 
rearrangements. The purpose of Theorem~\ref{main} and the following Theorem~\ref{main2} is to justify the intuitive conclusion and to clarify in which 
sense it is true.

The evolutionary solutions delivered by Theorem~\ref{main} only guarantee that, after going back from the coarse graining picture, the minimal energies 
for $\ell_\varnothing$, $\ell_r$, and $\ell_\varw$ will be close. However, for a discriminating protein to work correctly, it is important to have the 
right order of these energies: $U_r < U_\varnothing < U_\varw$. Therefore, we will consider the infinitesimal discrimination problem that probes the 
validity of this constraint by infinitesimally small deformation of solutions for the reduced discrimination problem at $\ell_1$.

More specifically, consider a nonzero vector $\varv$ on $L$ emanating from $\ell_1$ ($\varv \in T_{\ell_1} L$). This vector can be regarded as showing 
the direction from $\ell_r$ to $\ell_\varw$. These points are considered to be infinitesimally close to $\ell_1$, and $\ell_1$ has the same energy 
minimum as $\ell_0$. Therefore, for the phenotype $a$ to be a solution to the full discrimination problem, the displacement along $\varv$ with the 
fixed $a$ must increase the minimal energy. The boundary between solutions that respect this requirement and those that do not is made of $a$ such that 
there is no change in the minimal energy level in the direction spanned by $\varv$.

The second main result concerns this additional infinitesimal constraint on the order of the minima and is formulated as the following theorem.

\begin{theorem}
For a typical family of potentials $U$, solutions to the infinitesimal discrimination problem for ligands $\ell_0$ and $\ell_1$ and for a separating 
vector $\varv$ either do not exist, or a typical solution is located on a $(\dim A - 1)$-dimensional submanifold $\tilde {\normalfont \Upsilon}$ of $A$ 
and, if $\dim X > 0$, its minimum points $x_0$ and $x_1$ are different.
\label{main2}
\end{theorem}

In other words, the additional requirement of a correct order in energy minima does not qualitatively change the situation. It can make the set $\hat 
\Upsilon$ smaller, though ($\tilde \Upsilon \subset \hat \Upsilon$ in general).

\subsection{\normalsize Conformational changes and allosteric regulation}

Let us now look at how the development of a conformational change as a byproduct of a solution to the exquisite discrimination problem can help a 
development of allosteric regulation. By allosteric regulation we will understand the disruption of the initial ability to discriminate two ligands by 
binding of another ligand to a distant site of the molecule. 

Let us assume now that the protein in question can bind two different ligands: $\lambda$ and $\rho$. Therefore, the environmental variable takes the 
form $\ell = (\lambda,\rho)$ and it belongs to the space $L = \Lambda \times \Rho$, $\lambda \in \Lambda$, $\rho \in \Rho$. Let us denote, as before, 
the situation when $\lambda$ is bound by $\lambda_1$ and when it is not bound by $\lambda_0$. Likewise, we have $\rho_1$ and $\rho_0$ for the bound and 
free state of the ligand $\rho$. Note, that we assume, as before, that $\lambda_1$ in fact represents two ligands: $\lambda_r$ and $\lambda_\varw$. The 
protein discriminates these ligands. In contrast, $\rho_1$ is assumed to be a single ligand, which is bound by the protein without discrimination. 
Based on the theorems of the previous section we can expect a conformational change of the protein upon binding of $\lambda$, upon binding of $\rho$, 
upond binding of $\lambda$, when $\rho$ is already bound, and \emph{vice versa}.

Let us now assume in addition that the binding of $\lambda$ and $\rho$ is localized on the molecule in question and that it happens at different sites. 
Let us also assume that the sites are not to directly coupled. This can be expressed in the following way. Let $x_1$ be the degrees of freedom involved 
in the interaction with the ligand $\lambda$ (coordinates of atoms interacting with $\lambda$, for example), $x_2$ be the degrees of freedom involved 
in binding $\rho$, and $x_0$ be the residual degrees of freedom. We assume thus that $X = X_0 \times X_1 \times X_2$ with $x_0 \in X_0$, $x_1 \in X_1$, 
and $x_2 \in X_2$. Then the potential decomposes in this case as
	\begin{equation}
		U(x,\ell,a) = U_0(x_0,x_1,x_2,a) + U_1(x_1,\lambda,a) + U_2(x_2,\rho,a).
	\label{U012}
	\end{equation}

Let $a$ be a solution to the reduced exquisite discrimination problem (the reasoning is analogous for the infinitesimal problem) for $\ell_{00} = 
(\lambda_0,\rho_0)$ and $\ell_{10} = (\lambda_1,\rho_0)$, and define $\ell_{01} = 
(\lambda_0,\rho_1)$ with $\ell_{11} = (\lambda_1,\rho_1)$. Then the following result holds.

\begin{theorem}
Suppose that the protein changes its conformation during the switch from $\ell_{01}$ to $\ell_{11}$ (upon binding of $\lambda$ on the background of 
bound $\rho$). Then the situation described by
	\begin{equation}
		\min U_{\ell_{00},\,a} = \min U_{\ell_{10},\,a}\quad \textrm{and}\quad \min U_{\ell_{01},\,a} = \min U_{\ell_{11},\,a}
	\label{situation1}
	\end{equation}
is not structurally stable in the sense that it can be turned by an arbitrarily small perturbation of $U$ into situation described by
	\begin{equation}
		\min U_{\ell_{00},\,a} = \min U_{\ell_{10},\,a}\quad \textrm{but}\quad \min U_{\ell_{01},\,a} \neq \min U_{\ell_{11},\,a}.
	\label{situation2}
	\end{equation}
In the contrary, situation~(\ref{situation2}) is structurally stable in the sense that for any small enough perturbation of $U$ 
it cannot be turned into situation~(\ref{situation1}).
\label{main3}
\end{theorem}

In other words, situation~(\ref{situation2}) means that when $\rho$ is not bound, the protein performs the exquisite discrimination for $\lambda$, 
while when $\rho$ is bound, this ability is broken. Therefore, $\rho$ acts as an allosteric regulator for the exquisite discrimination of $\lambda$. 
The theorem asserts that such behaviour is typical. The condition of the theorem substantiantly uses the genericity of the conformational change upon 
binding provided by Theorem~\ref{main}.

\section{\large Discussion and biological interpretation}
\label{section-discussion}

Theorems~\ref{main}--\ref{main3} provide rigorous results in the limit of indistinguishable ligands ($\ell_\varw \to \ell_r$) and for the zero 
temperature approximation. In real systems, the difference between the right and the wrong ligands is finite, free energies of different states are 
allowed to be different provided that the correct order is preserved, and the temperature is positive. Going back from the mathematical idealization 
adopted above to physically meaningful models with finite differences and nonzero temperature blurs the rigor of the statements. An exclusion from a 
generic situation must be understood as not something impossible for practical observation but rather as something less probable than the generic case. 
The more the difference and the temperature the less strong the statement. This can be graphically demonstrated for the case of a nonzero difference 
between $\ell_r$ and $\ell_\varw$ (and between the energy levels $U_r$, $U_\varw$, and $U_\varnothing$) still assuming zero temperature. A solution to 
the exquisite discrimination problem in this case corresponds to a phenotype $a$ such that (\ref{U-condition}) holds. If we denote $\ell_0 = 
\ell_\varnothing$ and $\ell_1 = \ell_r$, the corresponding set $\hat \Upsilon$ provided by Theorem~\ref{main} (we will denote $\hat \Upsilon_r$) 
defines the border of phenotypes that respect $U_r < U_\varnothing$. In the same way, the analogous set $\hat \Upsilon_\varw$ defined for $\ell_0 = 
\ell_\varnothing$ and $\ell_1 = \ell_\varw$ marks the border of phenotypes that respect $U_\varnothing < U_\varw$. When $\ell_r$ becomes close to 
$\ell_\varw$, $\hat \Upsilon_r$ becomes close to $\hat \Upsilon_\varw$. From this it is clear that for $\ell_r \neq \ell_\varw$, phenotypes that solve 
the exquisite discrimination problem are situated between $\hat \Upsilon_r$ and $\hat \Upsilon_\varw$. This is schematically shown by the shaded region 
on Figure~\ref{figUpsilon}. In fact, this regions is a ``thick'' version of the codimension-1 submanifold $\tilde \Upsilon$ given by 
Theorem~\ref{main2} with an appropriately chosen direction $\varv$. Indeed, when $\ell_r$ approaches $\ell_\varw$ by some trajectory, the shaded region 
collapses to the submanifold that corresponds to $\varv$ such that $\varv$ is tangent to the trajectory. We see that a nonzero difference between the 
ligands makes the possible evolutionary solutions to their discrimination problem to occupy a spatial domain in the trait space $A$ rather than its 
infinitely thin codimension-1 submanifold. Yet, with sufficiently similar ligands (which is supposed by the \emph{exquisite} discrimination problem) 
they stay near such manifold.

\begin{figure}
\centering
	\begin{tikzpicture}
		\draw (0,0) -- (4,0) -- (4,4)node[below left]{$A$} -- (0,4) -- (0,0);
		\node at (0.3,3.7) {\textbf{(a)}};
		\begin{scope}
			\clip (0,0) -- (4,0) -- (4,4) -- (0,4) -- (0,0);
			\begin{scope}
				\clip (0,1) to[out=0,in=-110] (1,2) to[out=80,in=-180] (2,3)
					to[out=0,in=150] (3,2) to[out=-30,in=-135] (4,2) to (4,4) to (0,4)[loop];
				\begin{scope}
					\clip (0,1) to[out=30,in=-140] (1,2) to[out=50,in=-150] (2,3)
						to[out=30,in=120] (3,2) to[out=-60,in=-105] (4,2)
						to (4,0) to (0,0)[loop];
					\fill[gray!50!white] (0,0) rectangle (4,4);
				\end{scope}
			\end{scope}
			\draw[thick,color=blue!50!black] (0,1) to[out=0,in=-110] (1,2) to[out=80,in=-180] (2,3)
				to[out=0,in=150] (3,2) to[out=-30,in=-135] (4,2);
			\draw[thick,color=red!50!black] (0,1) to[out=30,in=-140] (1,2) to[out=50,in=-150] (2,3)
				to[out=30,in=120] (3,2) to[out=-60,in=-105] (4,2);
			\draw[thick,->] (1.3,2.7) -- (1,3) node[above]{\footnotesize \color{blue} $U_r < U_\varnothing$};
			\draw[thick,->] (1.45,2.55) -- (1.75,2.25) node[below]{\footnotesize \color{red} $U_\varnothing < U_\varw$};
			\draw (3.5,1.85) -- (3.5,2.2) node[above]{\footnotesize \color{blue} $\hat \Upsilon_r$};
			\draw (3.5,1.75) -- (3.5,1.5) node[below]{\footnotesize \color{red} $\hat \Upsilon_\varw$};
			\draw (0.5,1.3) -- (2,0.5) node[below]{\small solutions} -- (2.7,2.5);
			\coordinate (a) at (2.5,2.8);
			\node at (a) {$\bullet$};
			\foreach \p in {0,30,...,360} {
				\draw[->] (a) -- ++(\p:0.3);
			}
		\end{scope}

		\begin{scope}[xshift=6cm]
			\begin{scope}
				\clip (0,0) -- (4,0) -- (4,4) -- (0,4) -- (0,0);
				\begin{scope}
					\clip (0,1) to[out=0,in=-110] (1,2) to[out=80,in=-180] (2,3)
						to[out=0,in=150] (3,2) to[out=-30,in=-135] (4,2) to (4,4) to (0,4)[loop];
						\fill[gray!50!white] (0,0) rectangle (4,4);
				\end{scope}
				\draw[thick,color=blue!50!black] (0,1) to[out=0,in=-110] (1,2) to[out=80,in=-180] (2,3)
					to[out=0,in=150] (3,2) to[out=-30,in=-135] (4,2);
				\draw[thick,->] (1.3,2.7) -- (1,3) node[above]{\footnotesize \color{blue} $U_r < U_\varnothing$};
				\draw (3.5,1.85) -- (3.5,2.2) node[above]{\footnotesize \color{blue} $\hat \Upsilon_r$};
				\draw (0.5,2.5) -- (2,1) node[below]{\small solutions};
				\coordinate (a) at (2.5,3.5);
				\node at (a) {$\bullet$};
				\foreach \p in {0,30,...,360} {
					\draw[->] (a) -- ++(\p:0.3);
				}
			\end{scope}
			\draw (0,0) -- (4,0) -- (4,4)node[below left]{$A$} -- (0,4) -- (0,0);
			\node at (0.3,3.7) {\textbf{(b)}};
		\end{scope}
	\end{tikzpicture}
	
	\caption{Solutions to the zero-temperature exquisite discrimination problem with finite difference between $\ell_r$ and $\ell_\varw$ \textbf{(a)} 
		and to the simple binding problem for a single ligand $\ell_r$ \textbf{(b)} and their sensitivity to mutations (see the text). A bundle of 
		arrows illustrates mutations away from a given solution.}
	\label{figUpsilon}
\end{figure}
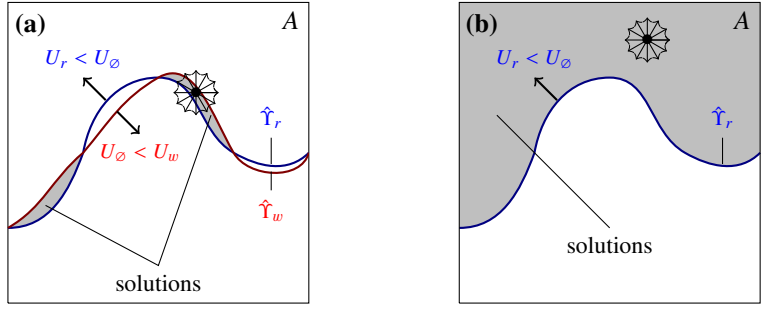

Another notion that is blurred in real systems is that of a large conformational change. In the idealized coarse grained mathematical model, any 
conformational change was interpreted as being large. The proven theorems do not provide any means to determine how large the conformational change is 
or what \emph{large} means in general. Such problems are typical for topological but not metric theorems. Addition of real physics on top of the bare 
topology in this problem (such as the limits on the stifness of the chemical bonds, assumption of a nonzero temperature, the value of the mutational 
effects of individual aminoacids, and so on) might help to destinguish between essential and nonessential changes in the discrimination ability of a 
protein and in its conformation.

Although the first part of the main result (Theorems~\ref{main} and \ref{main2}) can be shortly stated as \emph{discrimination requires a 
conformational change}, the statement would not be entirely correct. First, the \emph{requirement} must not be understood as direct causality. The 
correct interpretation is that most solutions to the discrimination problem will involve a conformational change. It implies that if a system performs 
discrimination and changes its conformation, it should not be surprising and no special explanation is required to this fact. In contrary, if a 
discriminating system does not show a conformational change, it is an indication on a special additional circumstances that may be of interest. Second, 
an application of the same theoretical approach to a protein that just binds to a ligand but not necessarily discriminates between similar ligands 
results in a conclusion that a conformational change is expected to accompany any binding in general.

Let us elaborate the latter statement. The part of the reasoning (in a simplified form) in the proof of Theorem~\ref{main} that involves a 
conformational change still holds in the case of a simple binding without discrimination. This means that we should expect a conformational change upon 
binding \emph{in general}, not only when a discrimination is performed. This general statement is very close to the classical induced fit scenario. The 
competing conformational selection hypothesis in its strict form, instead, represents a very special case (very special form of energy landscape), as 
it requires the conformation of the global minimum of free energy for the unbound state to be also a conformation of a local minimum for the bound 
state and \emph{vice versa}. This situation is not typical for smooth potentials. However, if the relevant conformations are themselves allowed to 
change upon binding, then the situation becomes as typical as the pure induced fit situation. In fact, the distinction between these two cases becomes 
irrelevant, as was already demonstrated in \cite{Rivoire2018} on a simple model. A similar conclusion was formulated in \cite{Hammes2009} based on 
biochemical arguments and experimental observations, and in a new vision of the protein binding proposed in \cite{Csermely2011}.

What discrimination does require is a an evolutionary fine tuning expressed in the dimension of the set of possible solutions. It is this fine tuning 
that brings about the potential for allosteric regulation (in the same sense as a discrimination causes a conformational change). If we combine the 
conclusion of Theorem~\ref{main3} with the above understanding of how such rigorous statements should be interpreted in application to real systems, we 
may conclude the following. \emph{Proteins that discriminate ligands are prone to allosteric regulation by another ligand at a different binding site. 
This sensitivity of the discrimination to the distant binding is associated with the conformational change during the primary binding.} Although this 
question was not studied in this work, we may also expect a wide sensitivity of such protein to mutations. Indeed, a mutation of an aminoacid is in 
some sense analogous to a local binding in its effect on the potential energy. Repeating for this case the reasoning about allosteric effects, we 
conclude that the ability to discriminate is broken by mutations in many sites. One can justify this assertion from a different perspective. Since the 
exquisite discrimination requires an evolutionary fine tuning, we can expect a mutation to break this tuning in a generic case. This is graphically 
represented on Figure~\ref{figUpsilon}. As a consequence, we expect a wide (in the spreading on the level of the primary sequence) mutational effect, 
when many mutations, however distant from the binding site, destroy the discrimination.

Note that the ability to bind a ligand \emph{without} an imposed discrimination problem is generically robust to most mutations. Indeed, the solutions 
to the binding problem for a single ligand lay in a half space of $A$ to one side of $\hat \Upsilon$ given by Theorem~\ref{main} for that ligand and 
the solvent. If a solution is situated deep in this region, it is expected to survive mutations in the sense that the resulting protein retains the 
binding ability (perhaps, with a weaker affinity, see Figure~\ref{figUpsilon}).

The modelling approach taken in this work is in the family of folding landscape models \cite{Bryngelson1995}. Looking at a protein through its (free) 
energy landscape is very natural from the point of view of physics and deserves more attention. The fact that such model supports the conclusion of 
\cite{Rivoire2018} is very important. It shows that an emergence of sophisticated properties of proteins and other biological heteropolymers, upon which 
substantial part of the complexity of life is built, can be attributed to a very simple evolutionary process: selection for a local property, that is 
the ability to discriminate between similar ligands. It is not difficult to imagine a selection process that optimizes this task. Furthermore, such 
ability very probably was required even back at the earliest times of abiogenesis or very early life.

\section{\large Proofs of Theorem~\ref{main}, Theorem~\ref{main2}, and Theorem~\ref{main3}}
\label{section-proof}

We imply in the following that all manifolds and functions (maps) are smooth. The main tool of the proof is the jet-bundle and the multijet 
transversality theorem that is a consequence of the Thom's transversality theorem. We will first recall some definitions and fix some notations.

\begin{definition}
Let $M$ and $N$ be two smooth manifolds and $S \subset N$ be a submanifold. Let $p$ be a point in $M$. A smooth function $f\colon M \to N$ is said to 
be \emph{transverse to $S$ at $p$}, if $df_p\, T_p M + T_{f(p)} S = T_{f(p)} N$, where $T_p M$ means the tangent space to $M$ at $p$ and $df_p$ is the 
differential of $f$ at $p$. $f$ is said to be \emph{transverse to $S$}, if it is transverse to $S$ at each point of $M$. This situation will be denoted 
by $f \pitchfork S$. Let $P \subset N$ be another submanifold. $S$ and $P$ are said to \emph{intersect transversely} (or simply to be 
\emph{transverse}), if $T_q S + T_q P = T_q N$ for each $q \in S \cap P$. This situation will be denoted by $S \pitchfork P$.
\end{definition}

\begin{definition}
The \emph{codimension} of a submanifold $N$ of a manifold $M$ is the number $\codim N = \dim M - \dim N$.
\end{definition}

If $S$ and $N$ are submanifolds of the same manifold and $N \pitchfork S$, then $\codim N\cap S = \codim N + \codim S$ (assuming $N \cap S \neq 
\varnothing$). If this number is negative, then $N \cap S = \varnothing$.

\begin{definition} Let $M$ be a smooth manifold. Two smooth functions $f$ and $g$ from $M$ to $\mathbb R$ are said to have \emph{$k$-th order contact 
at $p \in M$}, if in some coordinate chart around $p$ their values and all their partial derivatives up to order $k$ are equal at $p$. The relation of 
$k$-th order contact is independent of the coordinate chart and defines equivalent classes. The equivalent class of function $f$ by $k$-th order 
contact at $p$, denoted $[f]^k_p$, is called \emph{$k$-jet of $f$ at $p$}. Let $J^k(M,\mathbb R)_p$ be the set of all $k$-jets at $p$. The \emph{bundle 
of $k$-jets of functions on $M$} is the set $J^k(M,\mathbb R) = \coprod\limits_{p \in M} J^k(M,\mathbb R)_p$ with the projection $\pi_k\colon 
J^k(M,\mathbb R) \to M$, $[f]^k_p \mapsto p$ endowed with the differential structure lifted from $M$ by $\pi_k$. Every function $f\colon M \to \mathbb 
R$ generates a \emph{special section of the $k$-jet bundle} $j^k f\colon p \mapsto [f]^k_p$. \end{definition}

Note that jet bundles can be generalized to maps between arbitrary manifolds. Essentially, $k$-jets of functions represent an invariant notion of their 
Taylor polynomials truncated to order $k$. In the special case $k = 1$, the only one we will be interested in the following, the 1-jet bundle 
$J^1(M,\mathbb R)$ is naturally isomorphic to the product $\mathbb R \times T^*M$ (we denote this by $J^1(M,\mathbb R) \simeq \mathbb R \times T^*M$), 
where $T^*M$ is the cotangent bundle of $M$.

\begin{definition} An \emph{$s$-fold multijet bundle} $J^k_s(M,\mathbb R)$ is defined as follows. We denote
	\begin{equation}
		M^{(s)} = \{x \in M^s \colon \forall i,j, 
		1 \leqslant i < j \leqslant s \Rightarrow x_i \neq x_j \}.
	\end{equation}
It is a submanifold of $M^s$. Let $\pi_k$ be the bundle projection $J^k(M,\mathbb R) \to 
M$. Then $J^k_s(M,\mathbb R) = (\pi_k^{\times s})^{-1} (M^{(s)})$. It is a submanifold of $J^k(M,\mathbb R)^s$ and is a fibre bundle over $M^{(s)}$. 
Every function $f$ on $M$ generates its special section $j^k_s f$ by the rule $j^k_s f(x) = (j^k f(x_1),\ldots,j^k f(x_s))$. \end{definition}

Let us denote the \emph{diagonal} of the direct product $M^2$ as $\Delta M^2$. In the special case $s = 2$, the only one we will be interested 
in the following, $M^{(2)}$ has a simple representation: $M^{(2)} = M^2\setminus \Delta M^2$.

\begin{definition}
	Map $f\colon X \to Y$ is \emph{regular} at $p \in X$, if the rank of $df_p$ is maximal. If $f$ is not regular at 
	$p$, it is called \emph{singular} at $p$ and $p$ is called its \emph{critical point}. Map $f\colon X \to Y$ is an \emph{immersion}, if $df_p$
	is injective at every $p \in X$.
\end{definition}

We also need some known theorems.

\begin{theorem}[\cite{Golubitsky1973}, page 52, Theorem 4.4]
	Let $X$ and $Y$ be manifolds, $W \subset Y$ be a submanifold, and $f\colon X \to Y$ be a function and let $f \pitchfork W$. Then 
	$f^{-1}(W)$ is a submanifold of $X$. If in addition $f(X) \cap W \neq \varnothing$, then $\codim f^{-1}(W) = \codim W$.
\label{preimage}
\end{theorem}

\begin{theorem}[A special case of Mather's multijet transversality theorem, \cite{Golubitsky1973}, page 57, Theorem 4.13]
	Let $X$ be a manifold and $W$ be a submanifold of $J^k_s(X,\mathbb R)$. The subset of $C^\infty(X)$ constructed of functions $f$ that verify 
	$j^k_s f \pitchfork W$ is a residual set of $C^\infty(X)$ in the Whitney $C^\infty$ topology (for the definition, see \cite[p.~42]{Golubitsky1973}). 
	Moreover, if $W$ is compact, this subset is open. This theorem is called Thom's transversality theorem for $s = 1$ and thus $J^k_s(M,\mathbb R) = 
	J^k(M, \mathbb R)$, $j^k_s f = j^k f$.
\label{multijet}
\end{theorem}

We will first prove some lemmas.

\begin{lemma}
	Let $X$ be a compact manifold and $Y$ be a manifold, let $\pi\colon X\times Y \to Y$ be the projection to the second factor. Then for a typical 
	function $f\colon X\times Y \to \mathbb R$, the set of critical points of $f|_{\pi^{-1}(y)}$ is finite for any $y \in Y$.
\label{lemma-finite}
\end{lemma}

\begin{proof}

It is known that the subspace of functions that have only isolated points is of so called \emph{infinite codimension} (see \cite{Arnold1968} and 
\cite{Kurland1975} for definition and explanation and \cite{Arnold1968} and \cite{Tougeron1968} for the proof), a notion that is stronger than being 
typical (the latter implies the former). This property implies that for any $k$, the set of $k$-parameter families of functions with only isolated 
points is residual in the set of all $k$-parameter families. Function $f$ on $X\times Y$ is a ($\dim Y$)-parameter family of functions on $X \simeq 
\pi^{-1}(y)$. Thus, for any $y$, $f|_{\pi^{-1}(y)}$ has only isolated critical points. The finiteness of the number of critical points on every layer 
follows from the compactness of $X$.

\end{proof}

\begin{lemma}
	Let $X$ and $Y$ be manifolds and $S \subset X\times Y$ be a compact submanifold, let $\pi\colon X\times Y \to Y$ be the projection to the second 
	factor, let $\dim S = \dim Y -1$. Suppose that for each $y \in Y$, $\pi^{-1}(y)\cap S$ is finite. Then $\pi|_S$ is regular at a typical point of 
	$S$.
\label{lemma-regular}
\end{lemma}

\begin{proof}

The conclusion of the lemma is trivially true for $\dim Y = 1$ ($\dim S = 0$). Therefore, in the following, we will consider $\dim Y > 1$.

Recall that for a manifold $M$ ($\dim M = n$), a smooth association of a point $p \in M$ with a $k$-dimensional subspace in $T_p M$ is called a 
$k$-dimensional \emph{distribution} on $M$ \cite[\S 3]{simrus} (not to be confused with probability distributions). In other words, a distribution 
associates a tangent hyperplane with each point of a manifold. A distribution can be viewed as a subbundle of the tangent bundle. Another way to define 
a distribution is by defining a collection of (at least $k$) vector fields $V_i$ on $M$ that span the corresponding subspace of the distribution at 
each point. Finally, the same distribution can be defined by (at least $n - k$) differential 1-forms $\omega_j$ that annulate $V_i$: $\omega_j(V_i) = 
0$ for each $i$ and $j$. Any distribution can be defined in such way at least locally (in a neighbourhood of each point of $M$).

Now, the regularity of $\pi|_S$ at point $p \in S$ means that $T_p S \cap T_p \pi^{-1}(\pi(p)) = 0$. The association $p \mapsto T_p\pi^{-1}(\pi(p))$ 
defines a distribution $D$ on $X\times Y$, which is vertical towards the projection $\pi$ (it is mapped to the trivial distribution on $Y$ that 
associates $0 \in T_y Y$ with each point of $Y$) and has its layers as integral manifolds (the layers are tangent to the distribution at each point). 
In a local chart around $p \in X \times Y$ with coordinates $(x_\mu,y_\nu)$, the layers of the bundle $\pi$ are defined by the conditions $y_\nu = 
\mathrm{const}$, and thus the distribution is defined by 1-forms $\alpha_\nu = dy_\nu$, where $d$ is the exterior derivative.

The distribution $D$ on $X \times Y$ induces a distribution $D_S$ on $S$ in the following way. Let $\iota\colon S \to X\times Y$ be the inclusion of 
$S$. Then the collection $\{\omega_\nu\}$, $\omega_\nu = \iota^* \alpha_\nu$, defines a distribution on $S$, where $\iota^*$ is the pullback of 
differential forms induced by $\iota$. The dimension of $D_S$, however, can change from point to point depending on the degeneracy of $\{\omega_\nu\}$.

Choose a coordinate neighbourhood $\mathscr{O}$ around $s \in S$ in $S$ with local coordinates $s_\lambda$. Then locally we have $\omega_\nu = 
\sum\limits_\lambda \omega_{\nu\lambda} ds_\lambda$, where $\omega_{\nu\lambda} \in C^\infty(S)$. In these terms, the regularity of $\pi|_S$ at $s$ 
means that the rank of matrix $\omega_{\nu\lambda}$ is maximal at $s$ ($\rank \omega_{\nu\lambda}(s) = \dim S$, as $\dim S < \dim Y$).

Consider the sets
	\begin{align}
		&\Omega_0 = \{\sigma \in \mathscr{O} : \forall \nu_0,\, \omega_{\nu_0} = 0\} = \{\sigma \in \mathscr{O} : \rank \omega_{\nu\lambda}
			\leqslant 0\},\notag\\
		&\Omega_1 = \{\sigma \in \mathscr{O} : \forall \nu_0,\nu_1,\, \omega_{\nu_0} \wedge \omega_{\nu_1} = 0\} = \{\sigma \in \mathscr{O} : \rank 
			\omega_{\nu\lambda} \leqslant 1\},\notag\\
		&\ldots \notag\\
		&\Omega_k = \{\sigma \in \mathscr{O} : \forall \nu_0,\ldots,\nu_k,\, \omega_{\nu_0} \wedge \ldots \wedge \omega_{\nu_k} = 0\} = 
			\{\sigma \in \mathscr{O} : \rank \omega_{\nu\lambda} \leqslant k\},\label{Omega-k}\\
		&\ldots\notag
	\end{align}
where $\wedge$ is the exterior product of differential forms. Note that for each $k \leqslant l$, $\Omega_k \subset \Omega_l$. Note also that trivially 
$\Omega_k = \mathscr{O}$ for all $k \geqslant \dim S$. Consider also sets $\Theta_k$, where $\Theta_0 = \Omega_0$ and $\Theta_k = \Omega_k\setminus \Omega_{k-1}$ 
for $k > 0$. These sets define the points where rank of $\omega_{\nu\lambda}$ is equal to $k$.

Let us prove that $\Omega_k$ are closed and nowhere dense in $\mathscr{O}$ for $k < \dim S$, and thus $\Theta_{\dim S}$ is open and dense in 
$\mathscr{O}$. The closeness of all $\Omega_k$ follows from the fact that the defining equations in (\ref{Omega-k}) are equivalent to a finite set 
$\{F_m = 0\}$ of functional equations on $\omega_{\nu\lambda}$, where $F_m$ are homogeneous polynomials of $\omega_{\nu\lambda}$ of $k$-th order with 
coefficients from $\{-1,1\}$. Indeed, the equations in (\ref{Omega-k}) reflect nothing else but setting to zero all $k$-th minors of 
$\omega_{\nu\lambda}$. As $F_m$ are smooth, the set $\Omega_k = \{\sigma \in \mathscr{O} : F_m = 0\}$ is closed.

Now suppose that $\mathring \Omega_0 \neq \varnothing$, where $\mathring A$ is the interior of $A$. Then $D_S$ has constant dimension $\dim S$ in 
$\mathring \Omega_0$, and the set of equations $\omega_\nu(V) = 0$ has $\dim S$ linearly independent solutions. Choose one such $V$, a point $\sigma 
\in \mathring \Omega_0$, where $V_\sigma \neq 0$, a neighbourhood $\mathscr{O}_\sigma$ of $\sigma$, where $V \neq 0$ and rectifiable (which always 
exists by smoothness of $V$), the integral curve $\gamma$ of this vector field in $\mathscr{O}_\sigma$ that passes through $\sigma$, any $\sigma_1 \in 
\gamma$ different from $\sigma$, and denote $\tilde \gamma \subset \gamma$ the interval of $\gamma$ that connects $\sigma$ and $\sigma_1$. By 
necessity, for an arbitrary such $\sigma_1$ we have
	\begin{equation}
		\forall \nu\quad \intl{\tilde \gamma}{}\omega_\nu = 0,\quad \textrm{and thus}\quad \intl{\iota(\tilde \gamma)}{}dy_\nu = 
		y_\nu\big(\iota(\sigma_1)\big) - y_\nu\big(\iota(\sigma)\big) = 0.
	\end{equation}
The equality $y_\nu\big(\iota(\sigma_1)\big) = y_\nu\big(\iota(\sigma)\big)$ for all $\nu$ and $\sigma_1$ means that $\iota(\gamma) \subset 
\pi^{-1}(\pi(\sigma))$ and thus $S \cap \pi^{-1}(\pi(\sigma))$ is uncountable. This contradicts the premise, therefore $\mathring \Omega_0 = 
\varnothing$, which means that $\Omega_0 = \Theta_0$ is nowhere dense.

Repeat this reasoning in the inductive manner for $\Theta_k$, $0 < k < \dim S$. The only difference at each step is the number of independent vector 
fields that solve $\omega_\nu(V) = 0$, which is equal to $\dim S - k$. In the end of each step $\mathring \Theta_k = \varnothing$ (the proved 
expression) and $\mathring \Theta_{k-1} = \varnothing$ (the expression from the previous step) together imply $\mathring \Omega_k = \varnothing$. The 
induction chain breaks at $\Theta_{\dim S}$, since in this case the aforementioned equations have no nontrivial solutions, and thus $D_S$ is 
0-dimensional in $\Theta_{\dim S}$.

Now select a chart around every point of $S$ and then subselect a finite covering from this collection (which is possible by the compactness of $S$). 
Repeat the reasoning for all of them to get the conclusion of the lemma.

\end{proof}

Note that the requirements of compactness of manifolds and finiteness of $\pi^{-1}(y) \cap S$ are not essential. It is only essential for $\pi^{-1}(y) 
\cap S$ to be at most countable. But this level of generality brings about unneeded complications that are not relevant for the following.

\begin{proof}[Proof of Theorem~\ref{main}]

Let us call a \emph{presolution} to the reduced discrimination problem a phenotype $a$ such that the protein has equal in the energy level critical 
points of energy for $\ell_0$ and $\ell_1$, and not necessarily minima. It is clear that the proper solutions make a subset of the presolutions.

Consider the following diagram, associated with energy functions on $M$:
	\begin{equation}
		\xymatrix{
			\mathbb R \simeq \Delta \mathbb R^2\, \ar[r]^-{\iota_E} & \mathbb R^2 & & & L^2 & \{(\ell_0,\ell_1)\} 
				\ar[l]_-{\iota_L} \\
			& & J^1_2(M,\mathbb R) \ar[ul]_-{\pi_E} \ar[dl]_-{\pi_{T^*M}} \ar[r]^-{\pi_1} &
				M^{(2)} \ar@/^1.5pc/[l]^-{j^1_2 U} \ar[ur]^-{\pi_L} \ar[r]^-{\pi_A} \ar[dr]_-{\pi_X}
				& A^2 & \Delta A^2 \simeq A \ar[l]_-{\iota_A}\\
			P_X^2 \ar[r]^-{\iota_{P_X}} & (T^*M)^2 & & & X^2 & \Delta X^2 \simeq X \ar[l]_-{\iota_X}
		}
	\end{equation}

\noindent Here $\pi_1$ is the projection of $J^1_2(M,\mathbb R)$ as a bundle over $M^{(2)}$. $\pi_X$ is $(p_X \times p_X)|_{M^{(2)}}$, where $p_X$ 
is the natural projection of $M$ on $X$, analogously for $\pi_L$ and $\pi_A$. $\pi_E$ is the projection to pairs of energy values, associated with a 
multijet (it can be seen as $(p_{\mathbb R}\times p_{\mathbb R})|_{J^1_2(M,\mathbb R)}$, where $p_{\mathbb R}$ is the projection to the first factor of 
$\mathbb R \times T^*M \simeq J^1(M,\mathbb R)$). $\iota_X$, $\iota_L$, $\iota_A$, and $\iota_E$ are the obvious natural inclusions (embeddings).

Finally, $P_X$ and $\pi_{T^*M}$ are defined as follows. Let us consider again $J^1(M,\mathbb R)$ as $\mathbb R \times T^*M$ and let $p_{T^*M}$ be the 
natural projection on the second factor. In turn,
	\begin{equation}
		T^*M \simeq T^*X \oplus T^*(L\times A)
	\end{equation}

\noindent where $V\oplus W$ is the Whitney sum of two vector bundles $\pi_V:V \to B$, $\pi_W:W \to B$ over the same base $B$, i. e. the pullback from 
the following commutative diagram ($\iota_\Delta$ is the diagonal inclusion map)
	\begin{equation}
		\xymatrix{
			V\oplus W \ar[r] \ar[d] & V \times W \ar[d]^-{\pi_V \times \pi_W} \\
			B \ar[r]^-{\iota_\Delta} & B\times B 
		}
	\end{equation}
 
Let $0_X$ be the image of the 0-th section of $T^*X$ in $T^*X$. Then we define
	\begin{equation}
		P_X = 0_X \oplus T^*(L\times A),\quad \pi_{T^*M} = (p_{T^*M} \times p_{T^*M})|_{J^1_2(M,\mathbb R)},
	\end{equation}

\noindent and $\iota_{P_X}$ as the natural inclusion $P_X^2 \to (T^*M)^2$

Let us define
	\begin{align}
		&W_1 = (\pi_E)^{-1}(\im\iota_E) \cap (\pi_{T^*M})^{-1}(\im\iota_{P_X})
			\cap (\pi_L\circ\pi_1)^{-1}(\im\iota_L) \cap (\pi_A\circ\pi_1)^{-1}(\im\iota_A) \subset J^1_2(M,\mathbb R),\notag\\
		&W_2 = W_1 \cap (\pi_X\circ\pi_1)^{-1}(\im\iota_X) \subset J^1_2(M,\mathbb R),\notag\\
		&Y_1 = \im j^1_2 U \cap W_1 \subset J^1_2(M,\mathbb R),\notag\\
		&Y_2 = \im j^1_2 U \cap W_2 \subset J^1_2(M,\mathbb R),\notag\\
		&V_1 = \pi_1(Y_1) \subset M^{(2)},\\
		&V_2 = \pi_1(Y_2) \subset M^{(2)},\notag\\
		&\Upsilon_1 = \iota_A^{-1}(\pi_A(V_1)) \subset A,\notag\\
		&\Upsilon_2 = \iota_A^{-1}(\pi_A(V_2)) \subset A.\notag
	\end{align}

The meaning of these sets is the following. Space $J^1_2(M,\mathbb R)$ consists of pairs of jets over two at least somehow distinct points of $M$. The 
preimage of $\Delta \mathbb R^2$ defines pairs of jets that have the same energy values. The preimage of $P_X^2$ defines pairs of jets both of which 
have zero partial derivatives in $X$ (so, the corresponding points in $X$ are critical points for any representatives of these jets). The preimage of 
$(\ell_0,\ell_1)$ defines pairs of jets one of which is over $\ell_0$ and the other one is over $\ell_1$. The preimages of $\Delta A$ and $\Delta X$ 
define pairs of jets that have the same values of $a$ and $x$, correspondingly.

Therefore, $V_1$ corresponds to pairs of tuples $(x_0,\ell_0,a)$ and $(x_1,\ell_1,a)$ ($\ell_i$ are fixed to the values of the problem) such that 
functions $U(\cdot,\ell_i,a)$ have corresponding $x_i$ as critical points and $U(x_0,\ell_0,a) = U(x_1,\ell_1,a)$. $V_2$ corresponds to the same tuples 
but with the additional constraint $x_0 = x_1$. Accordingly, $\Upsilon_1$ corresponds to evolutionary presolutions of the reduced discrimination 
problem, while $\Upsilon_2$ corresponds to such presolutions where the critical points coincide.

By Theorem~\ref{multijet}, for a typical $U$ (from a residual set of all $U \in C^\infty(M)$), we have both $j^1_2 U \pitchfork W_1$ and $j^1_2 U 
\pitchfork W_2$. Indeed, the theorem guaranties that each of the condition is verified on a residual set. Therefore, they both are verified on a 
residual set, as an intersection of two residual sets is residual.

Sets $Y_i$ are compact submanifolds of $J^1_2(M,\mathbb R)$. Indeed, consider $W_i$ and $Y_i$ as subsets of $J^1(M,\mathbb R)^2 \supset J^1_2(M,\mathbb 
R)$. If on the diagram above we replace $J^1_2(M,\mathbb R)$ by $J^1(M,\mathbb R)^2$, $j^1_2 U$ by $j^1 U \times j^1 U$, $\pi_1$ by its nonrestricted 
version, all other projections of the form $\pi_\bullet$ by nonrestricted $p_\bullet \times p_\bullet$, and then repeat the construction of $Y_i$ in 
the same way, they will coincide with $Y_i$ constructed in the old way. Indeed, the only difference could be some additional points on the preimage of 
the diagonal $\pi_1^{-1}(\Delta M^2)$, but since $(p_L\times p_L)^{-1}(\im \iota_L) \cap \Delta M^2 = \varnothing$, we have $Y_i \cap \pi_1^{-1}(\Delta 
M^2) = \varnothing$, hence the equality. As $j^1 U \times j^1 U (M^2)$ is compact due to the compactness of $M^2$, $Y_i$ are compact, too. This 
property conserves upon restriction to $J^1_2(M,\mathbb R)$.

From the transversality properties and from $\pi_1 \circ j^1_2 U = \id_M$, by Theorem~\ref{preimage}, $V_i$ are submanifolds of $M^{(2)}$. These 
submanifolds are compact, too.

Note that
	\begin{align}
		&\codim (\pi_E)^{-1}(\im\iota_E) = 1,\quad \codim (\pi_{T^*M})^{-1}(\im\iota_{P_X}) = 2\dim X,\notag\\
		&\codim (\pi_L\circ\pi_1)^{-1}(\im \iota_L) = 2\dim L,\quad \codim (\pi_A\circ\pi_1)^{-1}(\im\iota_A) = \dim A,\\
		&\codim (\pi_X\circ\pi_1)^{-1}(\im\iota_X) = \dim X,\quad \dim M^{(2)} = 2(\dim X + \dim L + \dim A).\notag
	\end{align}

\noindent Therefore, by Theorem~\ref{preimage} and assuming $Y_i \neq \varnothing$,
	\begin{align}
		&\codim V_1 = \codim W_1 = 1 + 2\dim X + 2\dim L + \dim A,\notag\\
		&\codim V_2 = \codim W_2 = 1 + 3\dim X + 2\dim L + \dim A,\\
		&\dim V_1 = \dim A - 1,\quad \dim V_2 = \dim A - \dim X - 1.\notag
	\end{align}

We already see that if $\dim X > 0$ (the system is minimally flexible), then $\dim V_1 > \dim V_2$. As $V_2 \subset V_1$, we can conclude that points 
of $V_1$ that correspond to coincident critical values are not typical. More specifically, they form a submanifold of codimension $\dim X$. For 
instance, if $\dim A < \dim X + 1$, such points do not exist at all. Otherwise, if $\dim X > 0$, $V_2$ is a negligible subset in $V_1$ in the sense 
that it is nowhere dense and closed.

Note that by construction, $V_i \subset \pi_A^{-1}(\im \iota_A)\subset M^{(2)}$ and $\Upsilon_i$ can be understood as projections of $V_i$ on $A$ from 
$M^{(2)}$, for example as $\Upsilon_i = \pi (V_i)$, where $\pi = p_1 \circ \pi_A\colon M^{(2)} \to A$ and $p_1\colon A \times A \to A$ is the 
projection on the first factor. Unfortunately $\Upsilon_1$ and $\Upsilon_2$ are not in general submanifolds of $A$ due to generic singularities of the 
corresponding projection and generic self-crossings of the images of $\Upsilon_i$. We do not expect them to be even immersed manifolds. In fact, in a 
typical case, they form so called stratified sets of $A$. We will show, however, that typical points of $\Upsilon_1$ do form $(\dim A - 1)$-dimensional 
submanifolds with multiple connectness components.

Consider the fibre bundle $M \to L \times A$, where the projection is the natural projection of $M$ to the corresponding factor. Then $U$ can be 
viewed as a family of potentials $U_{\ell,a}$ on $X$ parametrized by $\ell$ and $a$. By Lemma~\ref{lemma-finite}, for typical $U$, all $U_{\ell,a}$ 
have only finite number of critical points for each pair $(\ell,a)$. Using the diagram
	\begin{equation}
		\xymatrix{
			P_X \ar[r] & T^*M & J^1(M,\mathbb R) \ar[r] \ar[l] & M \ar@/^1.5pc/[l]^-{j^1 U} \ar[r] & L & \{\ell_i\} \ar[l]
		}
	\label{diagram-li}
	\end{equation}
and the same reasoning as in the beginning of the proof, we conclude that for a typical $U$, the sets $V_{\ell_i}$ of all critical points in $X$ of 
$U_{\ell_i}$ at different $a$ (we will call $V_{\ell_i}$ the \emph{critical set} of $U_{\ell_i}$) constitute manifolds in $X\times \{\ell_i\}\times A 
\subset M$ and can be regarded as subsets of $X\times A \simeq X\times \{\ell_i\} \times A$.

Consider a fibre bundle with the natural projection $\xi_A \colon X\times A \to A$ with two functions $U_{\ell_0}$, $U_{\ell_1}$ on $X \times A$ that 
correspond to $U$ at different values of $\ell$. They too have finite number of critical point over every $a \in A$, since these functions are 
restrictions of $U$. If we consider the direct product of these fibre bundles with projection $\xi_A\times \xi_A \colon (X\times A)^2 \to A^2$, 
function $U_{\ell_0}\times U_{\ell_1}$ has $V_{\ell_0}\times V_{\ell_1}$ as its critical set. It is finite over any $(a_0,a_1) \in A^2$, too. 
Therefore, $V_1$ regarded as a submanifold of $X^2\times A\simeq X^2\times \Delta A^2$, is a submanifold of $V_{\ell_0}\times V_{\ell_1}$, and thus 
finite over any point $a \in A$. In other words, only finite number of points from $V_1$ are projected to any $a \in \Upsilon_1$. By 
Lemma~\ref{lemma-regular}, a typical point of $V_1$ (from an open dense subset) is projected regularly. Therefore, $\pi_A|_{V_1}$ is a local immersion 
in a neighbourhood of a typical point with finite preimage. Due to compactness of $V_1$, the set of points of change of the number of preimages and the 
intersection locus of the immersion in regular points are closed nowhere dense sets of $\Upsilon_1$. Therefore, typical points of $\Upsilon_1$ form 
open submanifold of $A$ of dimension $\dim A - 1$. $V_2$ projects to this submanifold as a $(\dim A - \dim X - 1)$-dimensional manifold and thus is 
closed and nowhere dense. After exclusion of all these meager points (singularities of projection, self-intersections, change of number of preimages, 
and $\Upsilon_2$), we are left with an open $(\dim A - 1)$-dimensional submanifold $\tilde \Upsilon$ of the phenotype space $A$, which is dense in 
$\Upsilon_1$.

Finally, let us return to the proper solution of the reduced discrimination problem, that is, when we consider only the parts of $V_1$ that correspond 
to the global minima and only the corresponding parts of $\Upsilon_1$. This will reduce $\Upsilon_1$ to a smaller subset $\hat \Upsilon$, but all the 
conclusions will hold for it, too. Typical points of $\hat \Upsilon$ form an open $(\dim A -1)$-dimensional submanifold of $A$ and the corresponding 
minimum points in $X$ over $\ell_0$ and $\ell_1$ do not coincide.

The only possible complication can come from situations when at least at one of $\ell_i$, $U_{\ell_i,a}$ has multiple global minima. Each part of $V_1$ 
that is projected to the corresponding $a$ guarantees only that each pair of a minimum point at $\ell_0$ and a minimum point at $\ell_1$ do not 
coincide, but it does not preclude a situation, when to the same point $a$, for example, two parts of $V_1$ are projected that correspond to pairs of 
minimum points at $\ell_0$ and $\ell_1$ of the form $(x_0,x_1)$ and $(x_1,x_2)$. In this case we would have two coinciding global minimum points for 
$\ell_0$ and $\ell_1$. However, as two members of the same pair correspond to the same value of the global minima, we have in this case
	\begin{equation}
		U(x_0,\ell_0,a) = U(x_1,\ell_1,a)\quad \textrm{and}\quad U(x_1,\ell_0,a) = U(x_2,\ell_1,a).
	\end{equation}
But all the minima must have the same value of energy, as they are global minima. Therefore, we must have
	\begin{equation}
		U(x_0,\ell_0,a) = U(x_1,\ell_0,a),\; \textrm{and thus,}\; U(x_1,\ell_0,a) = U(x_1,\ell_1,a).
	\end{equation}
It follows that such $a$ belongs to $\Upsilon_2$ and is not in the described submanifold of typical solutions to the reduced discrimination problem. 
This concludes the proof.

\end{proof}

\begin{proof}[Proof of Theorem~\ref{main2}]

Consider the following diagram
	\begin{equation}
	\xymatrix{
		P_X \ar[r] & T^*M & J^1(M,\mathbb R) \ar[r]^{\pi_{1,0}} \ar[l] & J^0(M,\mathbb R) \ar[r]^{\pi_0} &
			M \ar@/^1.5pc/[ll]^-{j^1 U} \ar@/_1.5pc/[l]_-{j^0 U}
	}
	\end{equation}
where $\pi_{1,0}$ is the natural projection $[f]^1_p \mapsto [f]^0_p$ (it forgets about tangency and keeps information only about intersections of 
function graphs) and $P_X$ is as in the proof of Theorem~\ref{main}. Consider the manifold (for generic $U$) $Y \subset J^1(M,\mathbb R)$ defined by 
the intersection of the preimage of $P_X$ and the image of $j^1 U$. Consider also, as before, the set $V$ of critical points of $U$, which is a manifold 
in $M$, $V = (j^1 U)^{-1}(Y)$. Consider now the projection $V_0$ of $Y$ to $J^0(M,\mathbb R)$. It is a manifold of $J^0(M,\mathbb R)$. Indeed, it is 
just $j^0 U(V)$ and $\im j^0 U$ is an embedding of $M$ (it is just the graph of $U$).

Note that $J^0(M,\mathbb R) \simeq \mathbb R \times M$. As before, we can assume that for every $\ell$ and $a$, $U_{\ell,a}$ has only finitely many 
critical points. Therefore, $V_0$ projected to $\mathbb R \times L \times A$ by the natural projection from $\mathbb R \times L \times A$ locally over 
a typical point $(\ell,a)$ looks like a finite set of sections (graphs) of the bundle $\mathbb R \times L \times A \to L \times A$. Nontypical points 
of $L \times A$ constitute a codimension-1 subset, which consists of either intersections of these manifolds or of the singularities of their 
projections. The set $W$ that corresponds to global minima of $U$ is a subset of this union of $(\dim L + \dim A)$-dimensional manifolds that locally 
in a typical point looks like a single such manifold that is regularly projected to $L \times A$.

For a typical $U$ and $\ell$, $U(\cdot, \ell, \cdot)\colon X \times A \to \mathbb R$ is a typical family of functions on $X$ parametrized by $A$. Thus, 
for a typical $a$, $U_{\ell_1,a}$ is Morse function. It has only separate nondegenerate critical points that are all mapped to different values. A 
codimension-1 bifurcations of a typical family (bifurcations that happen on a codimension-1 subset of $A$) include only fold bifurcations and equality 
of the function value for some two critical points. All other bifurcations have codimension greater than 2 and thus those of them that happen to be in 
$\hat \Upsilon$ form a meager set of $\hat \Upsilon$ (due to compactness of $A$, the complement to this set is open and dense in $A$). Therefore, we 
can consider that typical points of $\hat \Upsilon$ that possess the properties stated in Theorem~\ref{main} do not include these higher codimension 
bifurcations. Furthermore, possible codimension-1 bifurcations of the global minimum (and maximum) exclude fold bifurcations. Indeed, during a fold 
bifurcation, a critical point disappears in a collision with a nondegenerate critical point with different Morse index but for the global minimum it is 
impossible unless a third point becomes the new global minimum at the same time (this makes the bifurcation of codimension at least 2). Therefore, the 
only codimension-1 bifurcations of the global minimum are switches of the minimal points (coincidence of minimal values).

Let $\tilde A$ be the open and dense subset of $A$ formed by the complement to the set of bifurcations of codimension greater than 2 for global minima 
of the family $U_{\ell_1,a}$. Let us consider the intersection $W_{\ell_1}$ of the submanifold $\mathbb R \times \{\ell_1\} \times \tilde A \simeq 
\mathbb R \times \tilde A$ of $\mathbb R \times L \times A$ and $W$. From the previous paragraph we conclude that in some neighbourhood of each point 
of $W_{\ell_1}$, $W$ looks either like a graph of some smooth function $\phi_1\colon L \times A \to \mathbb R$ or like a graph of a continuous function 
$(\ell,a) \mapsto \min (\phi_1(\ell,a),\phi_2(\ell,a))$, where $\phi_1$ and $\phi_2$ are two smooth functions that are equal at the point in question. 
Let, for the first case, $s_1 = j^0 \phi_1$ be the corresponding section of the bundle $\pi\colon \mathbb R \times L \times A \to L \times A$, where 
$\pi$ is the natural projection. Let $s_1$ and $s_2$ be the sections corresponding to $\phi_1$ and $\phi_2$ of the second case. Let $I = \{1\}$ and $I 
= \{1,2\}$ for the first and the second cases, correspondingly. Let the local coordinates in $\mathbb R \times L \times A$ be $(E,\ell,a)$.

Let us locally define,for each $(\ell_1,a') \in W_{\ell_1}$, $a' \in \tilde A$, and a corresponding neighbourhood $\mathscr O_{(\ell_1,a')}$ that 
admits the aforementioned representation, functions on $\mathscr U_{a'} = \mathscr O_{(\ell_1,a')} \cap \tilde A$
	\begin{equation}
		f_i\colon a \mapsto dE_{s_i(\ell_1,a)}\,d(s_i)_{(\ell_1,a)}(\varv,0),\quad i \in I,
	\end{equation}
where by $(\varv,0)$ we understand the vector in $T_{\ell_1}L\oplus T_a A \simeq T_{(\ell_1,a)}(L \times A)$ that corresponds to $\varv$. These 
functions are smooth. Let us also define, for the same point, neighbourhood, and representation, the following piecewise smooth function $f$ on 
$\mathscr U_{a'}$. If the representation corresponds to $I = \{1\}$, then we define
	\begin{equation}
		f(a) = f_1(a).
	\end{equation}
If the representation corresponds to $I = \{1,2\}$, then we define
	\begin{equation}
		f(a) = \begin{dcases}
			f_1(a),& \phi_1(a) < \phi_2(a),\\
			f_2(a),& \phi_2(a) < \phi_1(a),\\
			\min(f_1(a),f_2(a)),& \phi_1(a) = \phi_2(a).
		\end{dcases}
	\end{equation}
It is not difficult to see that, by construction, the values of thus defined functions for any $a$ and any pair its intersecting neighbourhoods $U'_a$ 
and $U''_a$ (induced from any two $\mathscr O'_{\ell_1,a}$ and $\mathscr O''_{\ell_1,a}$) agree in $U'_a \cap U''_a$. Thus, function $f$ is globally 
defined on $\tilde A$. Moreover, we can use any its possible representation to study its local behaviour.

By construction, the sign of this function defines the sign of difference between $U_\varw$ and $U_r$ under an infinitesimal separation of $\ell_1$ to 
$\ell_r = \ell_1$ and $\ell_\varw$, when the latter moves along $\varv$. Let us extend $f$ on the whole $A$ by setting $f(a) = -1$ for $a \in 
A\setminus \tilde A$. Then the set $\Omega = \{a \in A : f(a) > 0\}$ is an open subset of $A$. Indeed, it means that the subset $\Kappa = \{a \in A : 
f(a) \leqslant 0\}$ is closed. To show this, let us choose some convergent sequence $\{a_n\}$ in $A$ such that $a_n \in \Kappa$ and $a_n \to a$. First 
note that $A \setminus \tilde A$ is trivially in $\Kappa$. Let $a \in \tilde A$. Starting from some number, all points of $a_n$ lay in some of 
$\mathscr U_a$ with one of the considered representations of $f$. If $I = \{1\}$ for this representation, by the smoothness of functions $\phi_1$ and 
$f_1$ and thus of $f$ in $\mathscr U_a$, we have $f(a_n) \to f(a)$ and $f(a) \leqslant 0$, thus $a \in \Kappa$. If, on the other hand, $I = \{1,2\}$ 
for the selected representation, the situation $f_{1,2}(a) > 0$ is impossible, and we must have either $f_{1,2}(a) \leqslant 0$, and thus $a$ is 
automatically in $\Kappa$, or one of $f_i(a)$ must be greater than 0 and the other one be smaller or equal to 0. Let us for concreteness assume $f_1(a) 
\leqslant 0 < f_2(a)$. But then, starting from sufficiently large number, we must have $f(a_n) = f_1(a)$ and thus $f(a) = f_1(a) \leqslant 0$, so again 
$a \in \Kappa$. Therefore, $\Kappa$ contains all its limit points and thus is closed, from which follows that $\Omega$ is open.

As an open subset of $A$, $\Omega$ is its ($\dim A$)-dimensional submanifold. Let $\check \Upsilon$ be the open $(\dim A - 1)$-dimensional sumbanifold 
of $A$ made of typical points granted by Theorem~\ref{main} ($\check \Upsilon$ is an open dense subset of $\hat \Upsilon$). Then, trivially, $\check 
\Upsilon \pitchfork \Omega$ and thus $\tilde \Upsilon = \check \Upsilon \cap \Omega$ is either empty or a $(\dim A - 1)$-dimensional submanifold of 
$A$. By construction, $\tilde \Upsilon$ consists of typical solutions to the infinitesimal discrimination problem with different minimal points.

\end{proof}

\begin{proof}[Proof of Theorem~\ref{main3}]

\begin{figure}
	\centering
	\begin{tikzpicture}
	    \tikzfading[name=fade in, inner color=transparent!100,outer color=transparent!0]
		\draw[->] (0,0) -- (6,0)node[below]{$X_0\times X_1$};
		\draw[->] (0,0) -- (0,4)node[left]{$X_2$};
		\draw[->] (0,0) -- (-0.5,-0.7)node[below]{$A$};
		\coordinate (a) at (4,2);
		\coordinate (a2) at (0,2);
		\coordinate (b) at (2,2);
		\coordinate (c) at (2,2.5);
		\begin{scope}
			\clip (2,2) circle (1cm);
			\draw[thick] (b) -- (c);
			\foreach \x in {1,1.2,...,3} {
				\draw[->-=(0.57-(\x-2)*(\x-2)/5)] (\x,0) -- (\x,4);
			}
			\fill [white, path fading=fade in] (2,2) circle (1cm);
		\end{scope}
		\node at (a) {$\bullet$};
		\node at (b) {$\bullet$};
		\node at (c) {$\bullet$};
		\draw (2,2) circle (1cm);
		\node at (2.6,0.75) {$\mathscr O$};
		\node at (a) [above right] {$(\hat x^{01},a)$};
		\draw[dashed] (b) -- (3,3.5)node[above right]{$(\hat x^{11},a)$};
		\draw[dashed] (c) -- (1,3.5)node[above]{$g^\epsilon(\hat x^{11},a)$};
		\draw[dotted] (a) -- (a2)node[left]{$\hat x^{01}_2,\hat x^{11}_2$};
	\end{tikzpicture}
		\caption{
			A possible deformation of $U_0$ to recover $\hat x^{01}_2 \neq \hat x^{11}_2$ (see the text). 
		}
\label{fig-deformation}
\end{figure}

Let $\hat x^{ij}$ denote the minimal points of, respectively, $U_{\ell_{ij},\,a}$ (for a typical value of $a$ the minimal points are unique and 
nondegenerate \cite[Propositions 6.3 and 6.13]{Golubitsky1973}). They can be represented as $\hat x^{ij}_{} = (\hat x^{ij}_0,\hat x^{ij}_1,\hat 
x^{ij}_2)$, $\hat x^{ij}_k \in X_k$. Let us also suppose that $\hat x^{01}_2 \neq \hat x^{11}_2$. If not, this can be recovered by an arbitrarily small 
perturbation of $U$ as follows. First recall that if $W$ is a manifold, $\mathscr O$ is any its open subset, $\mathscr C_1$ and $\mathscr C_2$ are any 
its closed subsets such that $\mathscr C_1 \cap \mathscr C_2 = \varnothing$, than there are smooth functions $\phi$ and $\psi$ on $W$ such that
	\begin{align}
		&\phi(p) > 0 \text{ for all } p \in \mathscr O\quad \text{and}\quad \phi(p) = 0 \text{ for all } p \notin \mathscr O;
			\label{function-phi} \\
		&\psi(p) = 0 \text{ for all } p \in \mathscr C_1,\quad \psi(p) = 1 \text{ for all } p \in \mathscr C_2,\quad \text{and}\quad 
			0 < \psi(p) < 1 \text{ for all } p \notin \mathscr C_1 \cup \mathscr C_2.
			\label{function-psi}
	\end{align}
An explicit construction of such functions can be found in \cite[Part 2, p. 13]{nestruev}. 

By the premise, we have $\hat x^{01} \neq \hat x^{11}$. Consider now a neighbourhood $\mathscr O$ of point $(\hat x^{11},a)$ in $X \times A$ such that 
$(\hat x^{01},a) \notin \mathscr O$. Choose an associated with $\mathscr O$ by (\ref{function-phi}) function $\phi$ and a nonzero in $\mathscr O$ 
rectifiable vector field $\varv$ tangent to $T_p X_2$ at each point $p \in X \times A$ (such vector field always exists in a small enough $\mathscr 
O$). Consider now the vector field $\varw = \phi \varv$ and the phase flow $g^t$ generated by $\varw$ (so $q(t) = g^t(p)$ is the solution to the 
differential equation $dq/dt = \varw_q$ with $q(0) = p$). The function $U_0^{(\epsilon)} = U^{}_0 \circ g^\epsilon$ provides the needed deformation of 
$U^{}_0$ (see Figure~\ref{fig-deformation}). The corresponding deformation $U^{(\epsilon)}_{} = U_0^{(\epsilon)} + U_1^{} + U_2^{}$ of $U$ tends to $U$ 
(in the Whitney $C^\infty$ topology) as $\epsilon \to 0$. Note also that $U^{(\epsilon)}$ still verifies (\ref{situation1}), but for any $\epsilon \neq 
0$, $\hat x^{01}_2 \neq \hat x^{11}_2$. From the other hand, if these points are different, this cannot be changed by an arbitrarily small perturbation 
of $U$, as the position of nondegenerate critical points of a function smoothly depends on this perturbation to a certain extent. Therefore, the 
inequality of $\hat x^{01}_2$ and $\hat x^{11}_2$ is typical for a typical $a$ that solves the reduced discrimination problem for $\ell_{01}$ and 
$\ell_{11}$.

Now let (\ref{situation1}) hold for $U$ and $\hat x^{01}_2 \neq \hat x^{11}_2$. Let us denote for brevity $N = X_2 \times \Rho \times A$. Let us also 
denote $p^{ij} = (\hat x^{ij}_2,\rho_2,a)$, $p^{ij} \in N$. Let $\mathscr C_1$ and $\mathscr C_2$ be two closed subset of $N$ such that $p^{11} \in 
\mathring {\mathscr C}_1$ (the interior of $\mathscr C_1$), $X_2 \times \{\rho_1\} \times A \subset \mathring {\mathscr C}_2$, $p^{01} \in \mathring 
{\mathscr C}_2$, and $\mathscr C_1 \cap \mathscr C_2 = \varnothing$ (such subsets always exist as $\hat x^{01}_2$ and $\hat x^{11}_2$ are different). 
Choose a function $\psi$ defined by $\psi \equiv 1$ in $\mathscr C_1$, $\psi \equiv 0$ in $\mathscr C_2$, as in (\ref{function-psi}). Consider a 
deformation $U_2^{(\epsilon)}$ of $U^{}_2$ in the following form $U_2^{(\epsilon)} = U^{}_2 + \epsilon \psi$ and the corresponding deformation of $U$ 
given by $U_{}^{(\epsilon)} = U^{}_0 + U^{}_1 + U_2^{(\epsilon)}$. For any number $\epsilon$ it has the same structure as in (\ref{U012}) and 
$U^{(\epsilon)} \to U$ (in the Whitney $C^\infty$ topology) as $\epsilon \to 0$. However, for an arbitrary small enough $\epsilon \neq 0$, it violates 
(\ref{situation1}) but verifies (\ref{situation2}). The robustness of situation (\ref{situation2}), in turn, again follows from the properties of 
nondegenerate critical points of functions.

\end{proof}

\section*{\large Acknowledgments}

The author thanks Olivier Rivoire for the formulation of the problem and for fruitful discussions, and Cl\'ement Nizak for carefully reading the 
manuscript and for his advices on making it more accessible to readers without strong mathematical background. This work was supported by ProTheoMics 
grant from Paris-Sciences-Lettres University and by ANR-17-CE30-0021-02 RBMPro grant from Agence Nationale de la Recherche.

\end{document}